\documentclass[a4paper,UKenglish,cleveref, autoref, thm-restate]{lipics-v2021}
\nolinenumbers

\usepackage{amsmath,amssymb,amsthm}

\usepackage[dvipsnames]{xcolor}

\usepackage{tikz}
\usepackage{tikz-cd}
\usepackage{tikzit}

\tikzstyle{white vertex}=[fill=white, draw=black, shape=circle, thick, inner sep=1pt, minimum size=6pt]
\tikzstyle{green vertex}=[fill=green, draw=black, shape=circle, thick, inner sep=1pt, minimum size=6pt]
\tikzstyle{cyan vertex}=[fill=cyan, draw=black, shape=rectangle, thick, inner sep=1pt, minimum size=6pt]
\tikzstyle{magenta vertex}=[fill={magenta!60}, draw=black, shape=isosceles triangle, thick, inner sep=1pt, minimum size=6pt, shape border rotate=90, isosceles triangle apex angle=60]
\tikzstyle{empty node}=[fill=white, shape=circle, inner sep=1pt]
\tikzstyle{small font vertex}=[fill=none, draw=none, shape=circle, font={{\scriptsize}}]
\tikzstyle{new style 0}=[fill=none, draw=none, shape=circle]
\tikzstyle{hypergraph edge}=[rounded corners, inner sep=0pt]

\tikzstyle{filled path}=[-, fill={blue!60}, thick, fill opacity=0.3]
\tikzstyle{thick edge}=[-, thick]
\tikzstyle{dashed path}=[-, dashed]
\tikzstyle{arrow path}=[->, dashed, thick]
\tikzstyle{hyperedge}=[-]
\tikzstyle{none filled path}=[-, draw=none, fill={green!60}, fill opacity=0.15]
\tikzstyle{arrow cyan}=[->, draw=cyan, thick]
\tikzstyle{arrow green}=[draw=green, ->, thick]
\tikzstyle{green thick edge}=[-, draw=green, thick]
\tikzstyle{cyan thick edge}=[-, draw=cyan, thick]
\tikzstyle{arrow solid}=[->]
\tikzstyle{cube face}=[-, fill={blue!60}, draw=none, opacity=0.2]
\tikzstyle{fill green}=[-, fill={green!60}, draw=black, fill opacity=0.2]
\tikzstyle{fill cyan}=[-, fill={cyan!60}, fill opacity=0.2]
\tikzstyle{fill magenta}=[-, fill={magenta!60}, fill opacity=0.2]
\tikzstyle{gray arrow}=[->, semithick, color=gray, >=angle 60]

\usepackage{quiver}
\usepackage{multicol}

\usepackage{tcolorbox}
\usepackage{todonotes}

\usepackage{cleveref}
\usepackage{mathpartir}

\usepackage{styles}

\title{A many-sorted epistemic logic for chromatic hypergraphs}

\author{\'Eric Goubault}{LIX, CNRS, École Polytechnique, IP-Paris, Palaiseau Cedex, France \and \url{https://www.lix.polytechnique.fr/~goubault/}}{goubault@lix.polytechnique.fr}{}{}

\author{Roman Kniazev}{LIX, CNRS, École Polytechnique, IP-Paris, Palaiseau Cedex, France\\ Université Paris-Saclay, ENS Paris-Saclay, CNRS, LMF, 91190, Gif-Sur-Yvette, France \and \url{https://www.lix.polytechnique.fr/Labo/Roman.KNIAZEV/}}{roman@kameronton.com}{}{}

\author{J\'er\'emy Ledent}{Université Paris Cité, CNRS, IRIF, F-75013, Paris, France \and \url{https://www.irif.fr/~ledent/}}{jeremy.ledent@irif.fr}{https://orcid.org/0000-0001-7375-4725}{}

\authorrunning{\'E. Goubault, R. Kniazev and J. Ledent} 

\Copyright{\'Eric Goubault, Roman Kniazev, J\'er\'emy Ledent} 

\ccsdesc[500]{Theory of computation~Modal and temporal logics}

\keywords{Modal logics, epistemic logics, multi-agent systems, hypergraphs} 

\category{} 

\relatedversion{} 

\EventEditors{} 
\EventNoEds{1} 
\EventLongTitle{} 
\EventShortTitle{} 
\EventAcronym{} 
\EventYear{2023} 
\EventDate{} 
\EventLocation{ } 
\EventLogo{} 
\SeriesVolume{} 
\ArticleNo{1} 

\begin{document}

\maketitle

\begin{abstract}
We propose a many-sorted modal logic for reasoning about knowledge in multi-agent systems.
Our logic introduces a clear distinction between participating agents and the environment.
This allows to express local properties of agents and global properties of worlds in a uniform way, as well as to talk about the presence or absence of agents in a world.
The logic subsumes the standard epistemic logic and is a conservative extension of it.
The semantics is given in chromatic hypergraphs, a generalization of chromatic simplicial complexes, which were recently used to model knowledge in distributed systems.
We show that the logic is sound and complete with respect to the intended semantics.
We also show a further connection of chromatic hypergraphs with neighborhood frames.
\end{abstract}

\section{Introduction}

Epistemic logic is a modal logic that allows formal reasoning about knowledge and belief in multi-agent systems.
At its core lies the \emph{knowledge operator} denoted $K_a \phi$, which means that ``agent~$a$ knows that the formula~$\phi$ holds'', or simply ``$a$ knows $\phi$'' for short.
Since the work of~\cite{hintikka:1962}, epistemic logic has expanded in many directions, studying various operators such as common knowledge~\cite{HalpernM90}, distributed knowledge~\cite{fagin}, as well as studying how knowledge evolves over time, as in public announcement logics~\cite{Gerbrandy1997ReasoningAI}, temporal epistemic logics~\cite{HalpernV89}, or dynamic epistemic logics~\cite{hvdetal.DEL:2007}.
A common feature of those approaches is that they rely on the classic ``possible worlds'' semantics of normal modal logics, based of Kripke structures.
Indeed, a model for multi-agent epistemic logic $\Sfive$ usually consists of a set of \emph{possible worlds}~$W$, and for each agent~$a$, an equivalence relation $\sim_a \,\subseteq W \times W$ called the \emph{indistinguishability relation} of agent~$a$.
The intended semantics of such a model is that an agent~$a$ \emph{knows} that a formula~$\phi$ is true, then $\phi$ is true in every possible world that is indistinguishable from the real world for that agent.
Formally, this means that the satisfaction relation is defined as follows, given a model~$M$ 
and a world~$w \in W$ (thought of as the ``real world''):
\begin{equation} \label{eqn:Ka-kripke}
M,w \models K_a \phi \quad \text{iff} \quad M,w' \models \phi \text{ for every world $w' \in W$ such that } w \sim_a w'
\end{equation}

A recent line of work~\cite{gandalf-journal, Ditmarsch2020KnowledgeAS, Ditmarsch21Wanted, GoubaultLR21kb4, Ditmarsch22Complete, GoubaultKLR23semisimplicial} has been developing a new notion of model for epistemic logic called \emph{simplicial models}.
This approach was closely inspired by connections with distributed computing, where simplicial complexes have been very successful in modeling various models of computation~\cite{herlihyetal:2013}.
Compared to Hintikka's possible worlds semantics, this new approach represents a shift in perspective.
Rather than focusing on the \emph{worlds} (a.k.a.\ global states, in distributed computing terms) as the primary object of study, we instead focus on the agents' \emph{points of view} about the world (a.k.a.\ local states).
A possible world can then be defined as a set of \emph{compatible} points of view, one for each agent.

\subparagraph*{Worlds and views.}

To illustrate the distinction between possible worlds, and local views about the possible worlds, we use a classic example originally from distributed computing~\cite{HerlihyR95}.
Since this paper is not concerned with distributed computing, we instead tell a story about a card game.
Assume there are three agents, $\cA = \{a, b, c\}$, and a deck of four cards, $\{1, 2, 3, 4\}$.
We deal one card to each agent, and keep the remaining card hidden, so that each agent only knows its own card.
Our goal is to model this (static) epistemic situation.

In the standard Kripke model semantics, the main step is to identify the possible worlds: each possible distribution of the cards constitutes a possible world.
There are 24 such distributions, which we can denote by $W = \{ 123, 124, 132, 134, 142, 143, 213, 214, \ldots \}$, where for example the world ``$123$'' denotes the situation where agent~$a$ (resp.\ $b$, $c$) has received card number~$1$ (resp.\ $2$, $3$).
To get a Kripke model, one must also define the indistinguishability relations for each agent.
For instance, we have $123 \sim_a 132$. Indeed, from the point of view of agent~$a$, who holds the same card number~$1$ in both of those worlds, these two worlds are indistinguishable.
One can check that $\sim_a$ defined in this way is indeed an equivalence relation, with four equivalence classes, and similarly for $\sim_b$ and $\sim_c$.

In simplicial models, however, the central notion is that of \emph{local view}.
From the point of view of agent~$a$, who sees only his own card, there are four possible situations: he can be given cards  $1$, $2$, $3$ or $4$. We call these the \emph{views} of~$a$, denoted by $V_a = \{1_a, 2_a, 3_a, 4_a\}$.
Similarly for the other two agents, we have $V_b = \{1_b, 2_b, 3_b, 4_b\}$ and $V_c = \{1_c, 2_c, 3_c, 4_c\}$.
In order to get a model, one must moreover define which of those views are \emph{compatible}.
Indeed, a possible world can now be seen as a set of views, one for each agent. But not every combination is allowed: in our example, $\{1_a, 1_b, 1_c\}$ is not a compatible set of views, because at most one agent can be given the card number~$1$.
As before, there are 24 sets of compatible views, corresponding to the 24 possible worlds: $\{1_a, 2_b, 3_c\}$, $\{1_a, 2_b, 4_c\}$, etc.

Surprisingly, shifting our focus from worlds to local views reveals an underlying geometric structure in this model.
Indeed, the structure described above, with a set of views and a $n$-ary compatibility relations between those views, is known in mathematics as an \emph{abstract simplicial complex}\footnote{Technically, we did not explicitly require the downward-closure property of a simplicial complex. We will come back to this later when we move on to hypergraph models.}.
Simplicial complexes provide a combinatorial description of topological spaces.
In the picture below, each local view is represented as a vertex, and each set of compatible views is represented as a triangle between the corresponding three vertices.
The resulting shape is that of a triangulated torus, represented on the right in 3D view, and on the left as a flattened view with some repeated vertices.
Notice that we use colors to indicate agent names, as we will do in the rest of the paper.

\begin{center}
	\scalebox{0.8}{
		\begin{tikzpicture}[scale=1.4,auto,>=stealth',cloudgray/.style={draw=black,thick,circle,fill={rgb:black,1;white,3},inner sep=1pt},cloud/.style={draw=black,thick,circle,fill=white,inner sep=1pt,minimum size=8pt}, cloudblack/.style={draw=black,thick,circle,fill=black,inner sep=1pt,font=\color{white}}]
			
			\foreach \i in {0,...,6} {
				\coordinate (a\i) at (\i,0) {};
				\coordinate (b\i) at (0.5+\i,0.85) {};
				\coordinate (c\i) at (1+\i,1.7) {};  
			};
			
			\draw[thick, draw=black, fill=blue, fill opacity=0.2]
			(a0) -- (c0) -- (c6) -- (a6) -- cycle;
			\draw[thick] (b0) -- (b6);
			\draw[thick] (b0) -- (a1);
			\draw[thick] (b6) -- (c5);
			\foreach \i [count=\j from 2] in {0,...,4} {
				\draw[thick] (c\i) -- (a\j);
			};
			\foreach \i in {1,...,5} {
				\draw[thick] (c\i) -- (a\i);
			};
			\node[cloudgray] at (a0) {$1_a$};
			\node[cloud] at (a1) {$4_b$};
			\node[cloudblack] at (a2) {$1_c$};
			\node[cloudgray] at (a3) {$4_a$};
			\node[cloud] at (a4) {$1_b$};
			\node[cloudblack] at (a5) {$4_c$};
			\node[cloudgray] at (a6) {$1_a$};
			\node[cloudblack] at (b0) {$3_c$};
			\node[cloudgray] at (b1) {$2_a$};
			\node[cloud] at (b2) {$3_b$};
			\node[cloudblack] at (b3) {$2_c$};
			\node[cloudgray] at (b4) {$3_a$};
			\node[cloud] at (b5) {$2_b$};
			\node[cloudblack] at (b6) {$3_c$};
			\node[cloud] at (c0) {$1_b$};
			\node[cloudblack] at (c1) {$4_c$};
			\node[cloudgray] at (c2) {$1_a$};
			\node[cloud] at (c3) {$4_b$};
			\node[cloudblack] at (c4) {$1_c$};
			\node[cloudgray] at (c5) {$4_a$};
			\node[cloud] at (c6) {$1_b$};
			
			\path[->] ($(c0)+(0,0.4)$) edge node {$A$} ($(c6)+(0,0.4)$)
			($(a0)-(0,0.4)$) edge node[swap] {$A$} ($(a6)-(0,0.4)$)
			($(a0)-(0.4,0)$) edge node {$B$} ($(c0)-(0.4,0)$)
			($(a6)+(0.4,0)$) edge node[swap] {$B$} ($(c6)+(0.4,0)$);
		\end{tikzpicture}
	}
	\hfill
	\includegraphics[scale=0.35]{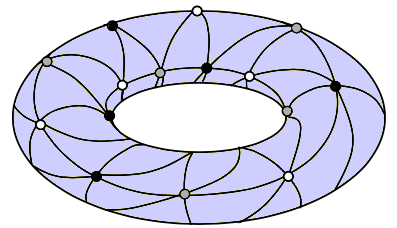}
\end{center}

This idea led to the definition of a \emph{pure simplicial model} in~\cite{gandalf-journal}.
Readers unfamiliar with simplicial complexes need not worry about technical details, as we will be using hypergraphs instead in this paper.
In the following, we assume the number of agents is $|\cA| = n+1$.

\begin{definition}[\cite{gandalf-journal}]
  \label{def:pure-simplicial-model}
  A pure simplicial model $M = (V, S, \chi, \ell)$ is given by:
  \begin{itemize}
    \item $(V, S)$ is a pure simplicial complex of dimension~$n$.
    \item $\chi : V \to \cA$ assigns agents to vertices, s.t.\ every simplex has vertices of different colors.
    \item $\ell : V \to 2^\AP$ assigns sets of atomic propositions to vertices.
  \end{itemize}
\end{definition}

With this notion of model, the analogue of condition (\ref{eqn:Ka-kripke}) becomes the following.
Note that now, the worlds $w, w'$ are facets of the simplicial model, i.e., sets of compatible vertices. Thus, the condition $a \in \chi(w \cap w')$ can be read as ``$w$ and $w'$ share an $a$-colored vertex'', that is, agent~$a$ has the same point of local view in both worlds.
\begin{equation} \label{eqn:Ka-simplicial}
	M,w \models K_a \phi \quad \text{iff} \quad M,w' \models \phi \text{ for every world $w'$ such that } a \in \chi(w \cap w')
\end{equation}

In \cref{def:pure-simplicial-model}, the requirement to be ``pure of dimension~$n$'' ensures that every agent is present in every world of the model.
While this is a standard assumption in the epistemic logic literature, it is often not the case in distributed computing.
Indeed, when we study computational models where processes may crash, one usually ends up with an impure simplicial model (see e.g.~\cite{Herlihy2001}).

\subparagraph*{Impure simplicial models.}

The idea of having a different set of agents (processes) in different possible worlds (executions) is ubiquitous in distributed computing.
This situation might occur when a process crashes during the execution of a protocol; or simply when the set of participating agents is not known in advance (say, a server concurrently answering requests from various clients).
In reference to the idea of crashed processes, and to be consistent with previous work on the topic~\cite{Ditmarsch21Wanted, GoubaultLR21kb4, Ditmarsch22Complete, GoubaultKLR23semisimplicial}, we will say that an agent can be ``alive'' or ``dead'' in a given world.
But note that for the time being, we only model static situations, so we could also say that agents can be ``present'' or ``absent''.

In the epistemic logic literature, the topic of non-participating agents has not been thoroughly studied.
It was briefly considered, e.g.\ in~\cite{fagin}, where it is called a ``nonrigid set of agents''.
This formalism is not very handy to work with, as it simply consists of extra data on top of the model indicating which agents are alive; and it is easy to circumvent the issue entirely by considering a special local state for crashed processes.
However, in simplicial models, it is quite natural and straightforward to model worlds with non-participating agents: we simply remove from \cref{def:pure-simplicial-model} the requirement that the model must be pure.
This simple idea led to a line of research on \emph{impure simplicial models}~\cite{Ditmarsch21Wanted, GoubaultLR21kb4, Ditmarsch22Complete, GoubaultKLR23semisimplicial}.

While it is clear which class of models we want to consider, we quickly run into issues when we try to define the semantics of epistemic logic formulas on these models.
The crux of the matter is that we have to decide how to define the satisfaction relation $w \models K_a \phi$, in a world $w$ where agent~$a$ is dead.
Two ways of dealing with this have been proposed.
\begin{itemize}
\item The first approach, called the three-valued semantics~\cite{Ditmarsch21Wanted, Ditmarsch22Complete}, claims that such a formula should be undefined in world~$w$.
Formulas can then be either true, false or undefined, hence the name ``three-valued''.
Defining when formulas are well-defined is not trivial, as knowledge operators can be nested and evaluating the satisfaction relation will explore the various possible worlds of the model. So one first needs to inductively define a judgment $w \bowtie \phi$, meaning that the formula~$\phi$ is well-defined in world~$w$; and then we can define the satisfaction relation $w \models \phi$ on top of it.
The resulting logic called $\Sfive^{\bowtie}$ is fully axiomatized in~\cite{Ditmarsch22Complete}.
It is a non-normal modal logic, where axiom $\mathbf{K} : K_a(\phi \rightarrow \psi) \rightarrow (K_a \phi \rightarrow K_a \psi)$, and even modus ponens, do not always hold.
However, it retains the axiom of truth, $\mathsf{T} : K_a\phi \rightarrow \phi$.
\item The second approach, called the two-valued semantics~\cite{GoubaultLR21kb4, GoubaultKLR23semisimplicial}, is obtained by closely following the correspondence with Kripke models.
As such, it yields a normal modal logic, where axiom $\mathsf{K}$ holds.
However, the axiom of truth is lost, and only a weaker version remains, saying that alive agents are truthful: $\aliveprop{a} \Rightarrow (K_a \phi \Rightarrow \phi)$.
The reason for this is that dead agents know every formula: when agent~$a$ is dead in world~$w$, the judgment $w \models K_a \phi$ is vacuously true, because there is no world~$w'$ satisfying the condition in (\ref{eqn:Ka-simplicial}).
The resulting logic is called $\KBfour$; however, small design choices in how we define the models can result in additional axioms, as  has been thoroughly investigated in~\cite{GoubaultKLR23semisimplicial}.
\end{itemize}

This situation is quite unsatisfactory, as both approaches seem to have pros and cons, and there is no obvious way to tell which one might turn out to be more useful in practice.
In this paper, we introduce a many-sorted logic that avoids entirely the problem of undefined formulas.
Rather than being a third proposal, it subsumes and unifies the previous two approaches.
But before we present the syntax of our logic, let us argue in favor of moving from simplicial complexes to hypergraphs.

\subparagraph*{From simplicial complexes to hypergraphs.}

In pure simplicial models (\cref{def:pure-simplicial-model}), the vertices represent points of view of individual agents, and the facets (a.k.a.\ maximal simplexes) represent the possible worlds.
There are also simplexes of lower dimension, but they do not seem to have a meaning. They are only here to preserve the geometric structure of the model: in order to have a triangle, we must also have the three edges of the triangle.

When we move on to impure simplicial models, we are now allowing worlds of different dimensions.
So there is no good reason why only the facets of the model should represent worlds.
In fact, as was shown in~\cite{GoubaultLR21kb4}, having only facets as worlds results in some dubious axioms.
The idea that all simplexes, not only the facets, could represent worlds was briefly discussed in~\cite{Ditmarsch2020KnowledgeAS}, and later studied in~\cite{Ditmarsch21Wanted, Ditmarsch22Complete}.
The idea was further generalized in~\cite{GoubaultKLR23semisimplicial}, where models are equipped with additional data (called a covering) allowing to explicitly say which simplexes are worlds or not.
Thus, models can be \emph{minimal} (only the facets are worlds), \emph{maximal} (all simplexes are worlds), or anything in-between.
For instance in \cref{fig:simplex-vs-hypergraph} (left), the model has only one facet (the blue triangle) but three worlds: $w_1$ is the triangle itself, where three agents are alive; $w_2$ is an edge, where only two agents are alive; and $w_3$ is a vertex, where only one agent is alive.

\begin{figure}[h]
\centering
\begin{tikzpicture}
	\begin{pgfonlayer}{nodelayer}
		\node [style=green vertex] (0) at (-5, 2) {};
		\node [style=cyan vertex] (1) at (-7, -1) {};
		\node [style=magenta vertex] (2) at (-3, -1) {};
		\node [style=green vertex] (27) at (5, 2) {};
		\node [style=cyan vertex] (28) at (3, -1) {};
		\node [style=magenta vertex] (29) at (7, -1) {};
		\node [style=none] (30) at (4.5, 2.25) {};
		\node [style=none] (31) at (3.5, -1.5) {};
		\node [style=none] (32) at (5.5, 1.75) {};
		\node [style=none] (33) at (2.25, -0.75) {};
		\node [style=none] (34) at (7.5, -1.5) {};
		\node [style=none] (35) at (4, 2.5) {};
		\node [style=none] (36) at (6, 2.5) {};
		\node [style=none] (37) at (2.25, -1.5) {};
		\node [style=none] (38) at (3, -0.5) {};
		\node [style=none] (39) at (2.5, -1) {};
		\node [style=none] (40) at (3, -1.5) {};
		\node [style=none] (41) at (3.5, -1) {};
		\node [style=none] (42) at (-5, 0) {};
		\node [style=none] (43) at (-5, 0) {$w_1$};
		\node [style=none] (44) at (-6.75, 1) {$w_2$};
		\node [style=none] (45) at (-7.75, -1.5) {$w_3$};
	\end{pgfonlayer}
	\begin{pgfonlayer}{edgelayer}
		\draw [style=hyperedge] (30.center)
			 to [in=75, out=45, looseness=1.50] (32.center)
			 to [in=15, out=-90, looseness=0.25] (31.center)
			 to [in=-105, out=-135, looseness=1.25] (33.center)
			 to [in=-165, out=90, looseness=0.25] cycle;
		\draw [style=hyperedge] (37.center)
			 to [in=-135, out=-60, looseness=0.50] (34.center)
			 to [in=-30, out=45, looseness=0.75] (36.center)
			 to [in=30, out=150] (35.center)
			 to [in=135, out=-150, looseness=0.75] cycle;
		\draw [style=hyperedge] (41.center)
			 to [in=0, out=90] (38.center)
			 to [in=90, out=-180] (39.center)
			 to [in=-180, out=-90] (40.center)
			 to [in=-90, out=0] cycle;
		\draw [style=filled path] (2.center)
			 to (1.center)
			 to (0.center)
			 to cycle;
	\end{pgfonlayer}
\end{tikzpicture}\hfill
\caption{A simplicial model with three worlds, and the equivalent hypergraph representation.}
\label{fig:simplex-vs-hypergraph}
\end{figure}
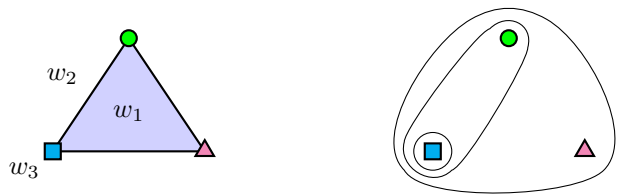

Still, the simplicial model depicted on the left has some simplexes that play no epistemic role (the two edges and two vertices with no label).
This introduces a strange discrepancy between simplexes that represent worlds, and simplexes that do not.
Our proposal, which is very mild from a technical standpoint, is to simply get rid of those extra simplexes.
The resulting structure is called a hypergraph.

\begin{definition} \label{def:simple-hypergraph}
A (simple) hypergraph $H$ is given by a pair $(V, E)$, where $V$ is a set of vertices, and $E \subseteq 2^V$ is a set of hyperedges (or just ``edges'', when clear from context).
\end{definition}

Note that the only difference between a hypergraph and a simplicial complex, is that the set~$E$ of hyperedges does not have to be downward-closed.
A hypergraph is depicted in \cref{fig:simplex-vs-hypergraph} (right).
It consists of three vertices $V = \{x, y, z\}$ (named from left to right), and three hyperedges $E = \{ \{x\}, \{x,y\}, \{x,y,z\} \}$.
Hyperedges are represented by a closed curve around the corresponding vertices, not unlike a Venn diagram.
By moving from simplicial complexes to hypergraphs, we seem to lose the geometric intuition behind simplicial complexes.
But it can be easily recovered by computing the downward closure of~$E$: if $(V, E)$ is a hypergraph, then $(V, \downarrow\! E)$ is a simplicial complex.

\begin{remark}
For readers familiar with the hierarchy of models introduced in~\cite{GoubaultKLR23semisimplicial}, hypergraph models fit nicely within that picture, as a strict subclass of epistemic covering models.
They are not as general as simplicial sets, as there cannot be complex connectivity between hyperedges.
They can be either minimal, or maximal, or in-between.
Simple hypergraphs as in \cref{def:simple-hypergraph} give rise to proper models, but as we will see later, we can also model non-proper behavior by allowing several hyperedges with the same sets of vertices.
\end{remark}

\subparagraph*{Many-sorted epistemic logic.}

We now describe the main contribution of the paper.
We propose a new syntax for epistemic logic formulas where agents can be dead or alive.
The central idea is to introduce several sorts of formulas\footnote{A familiar example of a many-sorted modal logic is CTL*, where the syntax is divided between state formulas and path formulas.}: \emph{world formulas} are to be interpreted in a world of the model, and \emph{agent formulas} are interpreted in a point of view of a particular agent (i.e.\ in a vertex, for hypergraph models).
We usually denote world formulas in capital letters $\Phi, \Psi$, and agent formulas in lowercase, with a subscript indicating the name of an agent $\phi_a, \psi_a$.
Thus, our logic has $|\cA|+1$ sorts, one for each agent $a \in \cA$, and an extra one for the world formulas.

A key observation is the following. Let us look again at the semantics of the knowledge operator for simplicial models, condition (\ref{eqn:Ka-simplicial}).
It says that an agent~$a$ knows~$\phi$ in a world~$w$, precisely when $\phi$ holds in every world $w'$ in which $a$ has the same point of view.
In fact, this definition does not refer to the ``real world''~$w$; it only refers to the point of view of~$a$ in this world!
This suggests that the knowledge operator $K_a$ should really be interpreted not in a world, but in a point of view of agent~$a$.
This gives us the following syntax for agent~$a$'s formulas, where $p_a$ ranges over atomic propositions concerning agent~$a$:
\[
	\varphi_a  ::=  p_a \mid 
	\neg\varphi \mid 
	\varphi \land \psi \mid
	\forallworld_a \Phi
\]
This gives us the first $|\cA|$ sorts of our logic, with one sort of agent formulas for each $a \in \cA$.
Note that agent formulas can only talk about one particular agent.
For instance, the expression $\forallworld_a \Phi \land \forallworld_b \Psi$ is not a syntactically valid formula.
Also note that $\Phi$ is a world formula. We now explain the syntax of world formulas.

Since agents can be alive or dead, we cannot talk about the knowledge of a specific agent in world formulas, or we will run into the issue of how to define the knowledge of a dead agent.
Instead, we introduce two new modal operators, $\existsagent_a$ and $\forallagent_a$, that can test whether an agent exists in this world.
As the names indicate, $\existsagent_a$ has an existential flavor, while $\forallagent_a$ is its universal counterpart.
However, note that they are not binders: $\existsagent_a$ should \textbf{not} be read as ``there exists an agent~$a$ such that''.
Rather, the intuitive meaning of those operators is the following:
\begin{itemize}
\item $\existsagent_a \phi_a$: ``there exists a point of view for agent~$a$ such that $\phi_a$ holds''.
\item $\forallagent_a \phi_a$: ``for every point of view of agent~$a$, $\phi_a$ holds''.
\end{itemize}

The syntax of world formulas is as follows, where $p_e$ denotes atomic propositions that do not talk about specific agents (`e' stands for \emph{environment} here).
\[
\Phi  ::=   p_e \mid
\neg\Phi \mid 
\Phi \land \Psi \mid
\existsagent_a \varphi_a \mid
\forallagent_a \varphi_a
\qquad \qquad \text{where } a \in \mathcal{A}
\]
Note that, unlike with agent formulas, we have $|\cA|$-many modal operators to choose from.
So, for example, the following world formula is syntactically valid: $\existsagent_a \forallworld_a \Phi \land \forallagent_b \forallworld_b \Psi$.
It is read ``there exists a point of view of agent~$a$ where $a$ knows $\Psi$, and for every point of view of agent~$b$, $b$ knows $\Psi$''.
Observe how whenever we want to talk about the knowledge of an agent, we are forced to explicitly quantify over the points of view of that agent.
This avoids entirely the question of ``undefined formulas'', where we had to make an arbitrary decision about the meaning of knowledge for dead agents.

\begin{remark}
	It might seem strange to have a modality $\forallagent_a$ quantifying over ``all points of view of agent~$a$'', when there can be at most one point of view per agent in a given world.
	First, note that when agent~$a$ is absent in a world, the operator $\forallagent_a$ is vacuously true, while $\existsagent_a$ is false.
	So these two operators do behave differently in hypergraph models.
	Secondly, one could consider an extension of hypergraph models where an agent can have multiple points of view about the world.
	We briefly explore this idea in~\cref{sec:neighborhood}, where we relate it to the neighborhood semantics of epistemic logic~\cite{pacuit2017neighborhood}.
\end{remark}

In some sense, our logic can be viewed as a \emph{refinement} of the usual knowledge operator~$K_a$ into two distinct operators:
$\forallagent_a \forallworld_a \Phi$ is the 2-valued semantics of~\cite{GoubaultLR21kb4, GoubaultKLR23semisimplicial}, which is vacuously true when agent~$a$ is dead; while $\existsagent_a \forallworld_a \Phi$ is closer to the 3-valued semantics of~\cite{Ditmarsch21Wanted, Ditmarsch22Complete}

\subparagraph*{Related work.}

As we already explained, this work is directly related to the line of work on simplicial models~\cite{gandalf-journal}, especially those that deal with impure simplicial complexes~\cite{Ditmarsch21Wanted, GoubaultLR21kb4, Ditmarsch22Complete, GoubaultKLR23semisimplicial}.
Recently, a single-sorted epistemic logic on hypergraphs was considered in \cite{DingLW23aggregative} to study weakly aggregative logics.
There, vertices of hypergraphs are not colored as they are interpreted as worlds, so these models lie in-between epistemic frames and neighborhood frames. 
A framework that uses adjoint modalities was studied in \cite{SadrzadehD09adjoint} in the context of epistemic modalities ``agent is uncertain about'' and ``agent has information that''.
In interpreted systems \cite{fagin}, epistemic frames are generated by explicitly modeling the local states of agents and global states of the environment. 
However, at the level of syntax, no difference is made between local properties of the agents, and global properties of the environment.

\subparagraph*{Plan of the paper.}

In \cref{sec:hyper-logic}, we start by describing the syntax of the logic $\twoCH$ and its semantics in chromatic hypergraphs.
We give an axiomatization of the logic in \cref{sec:prooftheory}, where we prove the completeness result in \cref{thm:completeness}.
In \cref{sec:equivalence}, we relate hypergraph models with partial epistemic models by showing an isomorphism of categories (\cref{prop:kripke-hypergraph}). 
We use this equivalence to formulate a translation from $\KBfour$-formulas into $\twoCH$-formulas, showing that the latter is a conservative extension of the former (\cref{thm:conservative-extension}).

\section{Two-level chromatic hypergraph logic $\twoCH$}
\label{sec:hyper-logic}

We introduce a many-sorted refinement of multi-agent epistemic logic $\KBfour$ studied in \cite{GoubaultLR21kb4}.
Our logic has $|\mathcal{A}|+1$ sorts of formulas: $|\mathcal{A}|$ sorts of \emph{agent} formulas, one per agent, and one sort of \emph{world} formulas.
Since this logic is intended to be interpreted on chromatic hypergraphs (see \cref{sec:semantics}), we name it the \emph{2-level Chromatic Hypergraph Logic}, $\twoCH$.
The intended interpretation of agent formulas is to describe local information that belongs to a point of view of a specific agent.
On the other hand, world formulas talk about properties of the environment, or world.
As we will see in Section \ref{sec:prooftheory}, this many-sorted logic, or two-level logic, embeds faithfully the logic $\KBfour$; but it also makes explicit (and not up to model interpretation as in \cite{GoubaultLR21kb4}) the various choices about how much agents can observe each other's presence or absence. 

\subsection{Syntax}

Fix a finite set~$\cA$ of agents.
For each $a \in \cA$, we have a set~$\AP_a$ of atomic propositions about agent~$a$.
We also have a set~$\AP_e$ of atomic propositions for the environment.
We use lowercase letters with subscripts $\phi_a, \psi_a, \ldots$ to denote agent formulas, and uppercase letters $\Phi, \Psi, \dots$ for world formulas.

\begin{definition}
The language of the logic $\twoCH$ is defined as follows.
For each agent $a \in \cA$, there is a sort of agent formulas generated by the following grammar:
\[
\varphi_a  ::=  p_a \mid 
\neg\varphi \mid 
\varphi \land \psi \mid
\existsworld_a \Phi
\qquad
\text{where } p_a \in \AP_a
\]
The sort of world formulas is generated by the following grammar:
\[
\Phi  ::=   p_e \mid
\neg\Phi \mid 
\Phi \land \Psi \mid
\existsagent_a \varphi_a
\qquad
\text{where } a \in \mathcal{A} \text{ and } p_e\in \AP_e
\]
\end{definition}

We will use standard propositional connectives like $\true$, $\lor$, $\Rightarrow$, defined as usual.
There are also dual modalities: $\forallworld_a \Phi := \neg\existsworld_a\neg \Phi$, and $\forallagent_a \phi_a := \neg\existsagent_a\neg \phi_a$.
Modalities are read as follows:
$\existsworld_a$ means ``agent $a$ considers possible that'',
$\forallworld_a$ means ``agent $a$ knows that'',
$\existsagent_a$ means ``there exists a point of view of agent $a$ such that'', and
$\forallagent_a$ means ``for all points of view of agent $a$''.
We call $\existsworld_a$ and $\existsagent_a$ \emph{existential} modalities and $\forallworld_a$ and $\forallagent_a$ \emph{universal} modalities.

\subsection{Semantics}
\label{sec:semantics}

A \emph{hypergraph} is a generalization of a graph, where instead of just edges between pairs of vertices, one has \emph{hyperedges} that can connect multiple vertices at once.
In the introduction, we defined \emph{simple} hypergraphs (\cref{def:simple-hypergraph}), to explain the proximity with simplicial complexes.
In fact, we will be slightly more general than that, and allow multiple hyperedges to have the same set of vertices; this will allow us to model non-proper behavior.
Namely, a (non-simple) hypergraph $H$ is a triple $(V, E, P)$, where $V$ is the set of vertices, $E$ is the set of hyperedges, and $P : E\to 2^V$ assigns to each hyperedge a set of vertices.

In the context of multi-agent systems, we need moreover to consider \emph{chromatic} hypergraphs, where each vertex is assigned an agent name.
This could be done by adding an extra piece of data $\chi : V \to \cA$, as in chromatic simplicial complexes~\cite{gandalf-journal}.
Instead, we tweak the definition a little bit in order to make explicit the set of vertices assigned to each individual agent.
This is similar to the definition of chromatic semi-simplicial sets in~\cite{GoubaultKLR23semisimplicial}.

\begin{definition}\label{def:chr-hypergraph}
    A chromatic hypergraph $H$ is a tuple
    $(E, \{V_a, \proj_a\}_{a\in\mathcal{A}})$, where:
    \begin{itemize}
    \item for all $a \in \mathcal{A}$, $V_a$ is the set of \emph{views} of agent $a$,
        \item $E$ is a set of hyperedges,
        \item for each agent $a\in \mathcal{A}$, $\proj_a: E\to V_a$ is a surjective partial function. Additionally, we require that for each $e\in E$, $\proj_a(e)$ is defined for at least one $a\in \mathcal{A}$.
    \end{itemize} 
\end{definition}

Indeed, defining $V=\bigcup_{a \in \mathcal{A}} V_a$ as the total set of  vertices of a chromatic hypergraph, and $P(e) = \{ \proj_a(e) \mid a \in \cA \} \subseteq 2^V$, we can view $H$ as a regular hypergraph. 
We will use the words \emph{view} and \emph{vertex} interchangeably, as well as \emph{world} and \emph{hyperedge}.
Given a world~$e \in E$, when $\proj_a(e)$ is undefined, we say that agent~$a$ is \emph{dead} in~$e$. Otherwise, $a$ is \emph{alive} in~$e$, and we call $\proj_a(e)$ the view of~$a$ in~$e$.
Moreover, when $\proj_a(e) = v$, we say that $v$ belongs to $e$, or that $e$ contains $v$, and occasionally write $v\in_a e$.
If a hyperedge consists of views $v_0,\dots v_{n-1}$ then we say that these views are \emph{compatible}.
Let us explain each condition imposed on $\proj_a$.
\begin{itemize}
	\item $\proj_a$ is surjective: every view belongs to at least one world.
	\item $\proj_a$ is partial: not all agents are required to be alive in a world. 
	\item $\proj_a$ is functional: every world contains at most one view of each agent.
	\item $\proj_a(e)$ is defined for at least one~$a \in \cA$: every world contains at least one alive agent.
\end{itemize}

\begin{remark}
In chromatic hypergraphs, views and worlds are of equal importance: they are both described explicitly in the sets~$E$ and $(V_a)_{a \in \cA}$.
This contrasts with Kripke structures, where the set of worlds is explicit, but the views are implicit in the indistinguishability relations.
Recent work on simplicial models has highlighted the importance of considering the views as a first-class notion.
Here, we take this idea one step further with our two-level syntax, where world formulas talk about properties of the world, and agent formulas talk about properties of an agent's point of view.
We also avoid choosing between \emph{local} vs.\ \emph{global} atomic propositions (in the sense of~\cite{gandalf-journal}): agent formulas can talk about (local) properties of agent $a \in \cA$ using atoms in~$\AP_a$, and world formulas can talk about (global) properties of the world using atoms in~$\AP_e$.
\end{remark}

From now on, we sometimes omit the adjective ``chromatic'' when clear from context.

\begin{example}\label{ex:hypegraphs}
    Three examples of chromatic hypergraphs are depicted in~\cref{fig:hyp-ex}.
    The three agents $a, b, c$ are represented as colored shapes $\agenta$, $\agentb$, $\agentc$, respectively.
    In all three hypergraphs, there is only one view per agent.

    On the leftmost figure, all possible combinations of views are compatible, giving 7 hyperedges.
    Intuitively, in this situation, the agents do not know whether other agents exist.
    A scenario like this is standard in distributed computing, where agents are processes, and they do not know whether other processes are concurrently running.

    The hypergraph depicted in the middle has a single hyperedge containing the three views.
    In this situation, all agents know that everyone is alive, as it is the only possible world.
    This represents a scenario where every agent has guarantees that the other two agents are running.

    The rightmost figure is a hypergraph with three hyperedges, each of which contains two views.
    It represents a situation where the points of views are pairwise compatible, but not all three of them are compatible.
    That is, there is no possible world that realizes all three of them at once.
    This could model a scenario where each agent receives a message from one of the other two agents, but they do not know who sent the message.
    Another interpretation might be in a quantum setting, as an example of \emph{contextuality}~\cite{AbramskyBKLM15}.

	\begin{figure}[htp]
	\centering
	\begin{tikzpicture}
	\begin{pgfonlayer}{nodelayer}
		\node [style=green vertex] (0) at (0, 2) {};
		\node [style=cyan vertex] (1) at (-2, -1) {};
		\node [style=magenta vertex] (2) at (2, -1) {};
		\node [style=none] (3) at (2, -2) {};
		\node [style=none] (4) at (0, 3) {};
		\node [style=none] (5) at (3, -1) {};
		\node [style=none] (6) at (-1, 2) {};
		\node [style=none] (7) at (0, 3) {};
		\node [style=none] (8) at (-2, -2) {};
		\node [style=none] (9) at (1, 2) {};
		\node [style=none] (10) at (-3, -1) {};
		\node [style=none] (11) at (2.5, -1.5) {};
		\node [style=none] (12) at (-2.5, -0.5) {};
		\node [style=none] (13) at (2.5, -0.5) {};
		\node [style=none] (14) at (-2.5, -1.5) {};
		\node [style=none] (15) at (3, -2) {};
		\node [style=none] (16) at (-2, 3) {};
		\node [style=none] (17) at (2, 3) {};
		\node [style=none] (18) at (-3, -2) {};
		\node [style=none] (19) at (0, 2.5) {};
		\node [style=none] (20) at (-0.5, 2) {};
		\node [style=none] (21) at (0, 1.5) {};
		\node [style=none] (22) at (0.5, 2) {};
		\node [style=none] (23) at (-2, -0.5) {};
		\node [style=none] (24) at (-2.5, -1) {};
		\node [style=none] (25) at (-2, -1.5) {};
		\node [style=none] (26) at (-1.5, -1) {};
		\node [style=none] (27) at (2, -0.5) {};
		\node [style=none] (28) at (1.5, -1) {};
		\node [style=none] (29) at (2, -1.5) {};
		\node [style=none] (30) at (2.5, -1) {};
	\end{pgfonlayer}
	\begin{pgfonlayer}{edgelayer}
		\draw [style=hyperedge] (3.center)
			 to [in=-75, out=-15] (5.center)
			 to [in=-15, out=105, looseness=0.25] (4.center)
			 to [in=105, out=165] (6.center)
			 to [in=165, out=-75, looseness=0.25] cycle;
		\draw [style=hyperedge] (7.center)
			 to [in=75, out=15] (9.center)
			 to [in=15, out=-105, looseness=0.25] (8.center)
			 to [in=-105, out=-165] (10.center)
			 to [in=-165, out=75, looseness=0.25] cycle;
		\draw [style=hyperedge] (11.center)
			 to [in=-45, out=45, looseness=1.25] (13.center)
			 to [in=45, out=135, looseness=0.50] (12.center)
			 to [in=135, out=-135, looseness=1.25] (14.center)
			 to [in=-135, out=-45, looseness=0.50] cycle;
		\draw [style=hyperedge] (18.center)
			 to [in=-135, out=-45, looseness=0.75] (15.center)
			 to [in=-30, out=45, looseness=0.75] (17.center)
			 to [in=30, out=150] (16.center)
			 to [in=135, out=-150, looseness=0.75] cycle;
		\draw [style=hyperedge] (22.center)
			 to [in=0, out=90] (19.center)
			 to [in=90, out=-180] (20.center)
			 to [in=-180, out=-90] (21.center)
			 to [in=-90, out=0] cycle;
		\draw [style=hyperedge] (26.center)
			 to [in=0, out=90] (23.center)
			 to [in=90, out=-180] (24.center)
			 to [in=-180, out=-90] (25.center)
			 to [in=-90, out=0] cycle;
		\draw [style=hyperedge] (30.center)
			 to [in=0, out=90] (27.center)
			 to [in=90, out=-180] (28.center)
			 to [in=-180, out=-90] (29.center)
			 to [in=-90, out=0] cycle;
	\end{pgfonlayer}
\end{tikzpicture}\hfill
	\begin{tikzpicture}
	\begin{pgfonlayer}{nodelayer}
		\node [style=green vertex] (0) at (0, 2) {};
		\node [style=cyan vertex] (1) at (-2, -1) {};
		\node [style=magenta vertex] (2) at (2, -1) {};
		\node [style=none] (15) at (3, -2) {};
		\node [style=none] (16) at (-2, 3) {};
		\node [style=none] (17) at (2, 3) {};
		\node [style=none] (18) at (-3, -2) {};
	\end{pgfonlayer}
	\begin{pgfonlayer}{edgelayer}
		\draw [style=hyperedge] (18.center)
			 to [in=-135, out=-45, looseness=0.75] (15.center)
			 to [in=-30, out=45, looseness=0.75] (17.center)
			 to [in=30, out=150] (16.center)
			 to [in=135, out=-150, looseness=0.75] cycle;
	\end{pgfonlayer}
\end{tikzpicture}\hfill
	\begin{tikzpicture}
	\begin{pgfonlayer}{nodelayer}
		\node [style=green vertex] (0) at (0, 2) {};
		\node [style=cyan vertex] (1) at (-2, -1) {};
		\node [style=magenta vertex] (2) at (2, -1) {};
		\node [style=none] (3) at (2, -2) {};
		\node [style=none] (4) at (0, 3) {};
		\node [style=none] (5) at (3, -1) {};
		\node [style=none] (6) at (-1, 2) {};
		\node [style=none] (7) at (0, 3) {};
		\node [style=none] (8) at (-2, -2) {};
		\node [style=none] (9) at (1, 2) {};
		\node [style=none] (10) at (-3, -1) {};
		\node [style=none] (11) at (2.5, -1.5) {};
		\node [style=none] (12) at (-2.5, -0.5) {};
		\node [style=none] (13) at (2.5, -0.5) {};
		\node [style=none] (14) at (-2.5, -1.5) {};
	\end{pgfonlayer}
	\begin{pgfonlayer}{edgelayer}
		\draw [style=hyperedge] (3.center)
			 to [in=-75, out=-15] (5.center)
			 to [in=-15, out=105, looseness=0.25] (4.center)
			 to [in=105, out=165] (6.center)
			 to [in=165, out=-75, looseness=0.25] cycle;
		\draw [style=hyperedge] (7.center)
			 to [in=75, out=15] (9.center)
			 to [in=15, out=-105, looseness=0.25] (8.center)
			 to [in=-105, out=-165] (10.center)
			 to [in=-165, out=75, looseness=0.25] cycle;
		\draw [style=hyperedge] (11.center)
			 to [in=-45, out=45, looseness=1.25] (13.center)
			 to [in=45, out=135, looseness=0.50] (12.center)
			 to [in=135, out=-135, looseness=1.25] (14.center)
			 to [in=-135, out=-45, looseness=0.50] cycle;
	\end{pgfonlayer}
\end{tikzpicture}
	\caption{Three examples of chromatic hypergraphs}
	\label{fig:hyp-ex}
	\end{figure}
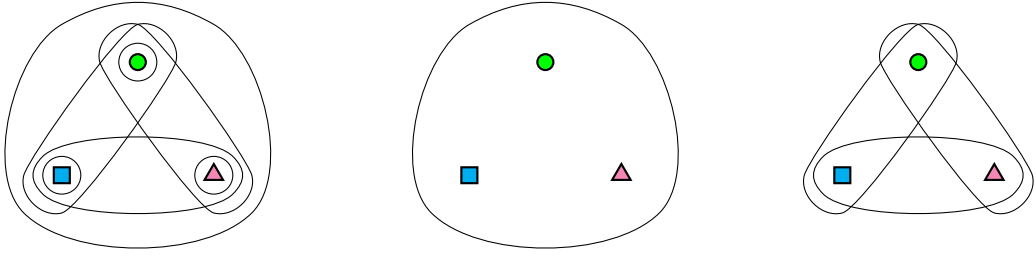
\end{example}

\begin{definition}
    A chromatic hypergraph model is a tuple $(H, \{\ell_a\}_{a\in \mathcal{A}}, \ell_e)$, where $H$ is a chromatic hypergraph, and $\ell_a: \AP_a\to P(V_a)$, $\ell_e: \AP_e \to P(E)$ are valuation functions.
\end{definition}

We can now define the semantics of $\twoCH$ formulas with respect to chromatic hypergraph models.
Given a hypergraph model $H$, the satisfaction relations are defined for every sort by mutual induction.
As expected, world formulas are interpreted in a world $e \in E$, and agent formulas are interpreted in a point of view $v \in V_a$ of that agent.

\noindent
\begin{minipage}{0.48\textwidth}
\small
\[
\begin{array}{lcl}
	H, v \models_a p_a &  \text{iff} & v \in \ell_a(p_a) \\
	H, v \models_a \neg\varphi & \text{iff} & H, v \not\models_a \varphi \\
	H, v \models_a \varphi\land\psi & \text{iff} & H, v \models_a \varphi  \text{ and } H, v \models_a \psi  \\
	H, v \models_a \existsworld_a \Phi  & \text{iff} & H, e \models_e \Phi \text{ for some } e \in E \\ & & \text{ such that } \proj_a(e)=v \\
\end{array}
\]
\end{minipage}
\hfill
\begin{minipage}{0.48\textwidth}
\small
\[
\begin{array}{lcl}
	H, e \models_e p_e & \text{iff} & e \in \ell_e(p_e) \\
	H, e \models_e \neg\Phi & \text{iff} & H, e \not\models_e \Phi \\
	H, e \models_e \Phi\land\Psi & \text{iff} & H, e \models_e \Phi  \text{ and } H, e \models_e \Psi  \\
	H, e \models_e \existsagent_a \varphi  & \text{iff} & H, v \models_a \phi \text{ for some } v \in V_a \\ & & \text{ such that } \proj_a(e)=v \\
\end{array}
\]
\end{minipage}

\subsection{Examples}

Let us illustrate the semantics of $\twoCH$ on a few examples.

\begin{example}
    Consider again the three hypergraphs from \Cref{ex:hypegraphs}, which we denote by $H_1, H_2, H_3$, from left to right.
    For now, we do not worry about atomic propositions: we just need to assume that the sets $\AP_\ast$ are non-empty, in order to get the constant $\true = p \lor \neg p$.
    Also, recall that the agents are depicted as $a = \agenta$, $b = \agentb$, and $c = \agentc$.
    
    In the model $H_1$, let us consider the hyperedge $e = \{ \agenta, \agentb \}$.
    Then we have $H_1, e \models \existsagent_a \true$ and $H_1, e \models \existsagent_b \true$, but $H_1, e \not\models \existsagent_c \true$.
    Indeed, the edge $e$ does not contain a point of view of agent~$c$, i.e., $\proj_c(e)$ is undefined.
    In fact, the world formula ``$\existsagent_a \true$'' is satisfied exactly in the worlds where agent~$a$ is alive.
	So let us write $\aliveprop{a} := \existsagent_a \true$, so that $H, e \models \aliveprop{a}$ iff $e$ contains a point of view of~$a$, that is, iff $\proj_a(e)$ is defined.
    
    We can now talk about whether agents know that other agents are alive.
    Let $v = \agenta$ be the (unique) point of view of~$a$ in the model.
    In model~$H_1$, all combinations of alive and dead agents are possible.
    So as expected, $H_1, v \models \neg \forallworld_a \aliveprop{b}$.
    In model~$H_2$ however, $a$ knows that all agents are alive, since this is the only possible world: $H_2, v \models \forallworld_a(\aliveprop{b} \land \aliveprop{c})$.
    Finally in model~$H_3$ the situation is more complicated: $a$ knows that another agent is alive, but does not know which one. Thus, $H_3, v \models \forallworld_a(\aliveprop{b} \lor \aliveprop{c}) \land \neg \forallworld_a \aliveprop{b} \land \neg \forallworld_a \aliveprop{c}$.
\end{example}

\begin{example}[$2$-agent binary input model with solo executions]
As an example where atomic propositions play a role, we consider the situation where two agents are given a binary input value, either $0$ or $1$.
Moreover, there can be solo executions (a.k.a.\ initial crash failures, in distributed computing), so that agents do not know if they are running alone or not.
As before, the agents are depicted as $a = \agenta$ and $b = \agentb$.
The sets of atomic propositions are $\AP_a = \{ 0_a, 1_a \}$, $\AP_b = \{ 0_b, 1_b \}$, and $\AP_e = \{ \mathsf{solo} \}$.
The agent atomic propositions hold in the points of view indicated on the picture below.
The environment atomic proposition $\mathsf{solo}$ holds in the four singleton hyperedges where only one agent is alive.

\begin{center}
	\begin{tikzpicture}
	\begin{pgfonlayer}{nodelayer}
		\node [style=green vertex] (0) at (-1.5, 1.5) {};
		\node [style=cyan vertex] (1) at (1.5, 1.5) {};
		\node [style=none] (3) at (-2.5, 1.5) {};
		\node [style=none] (4) at (-1.5, 2.25) {};
		\node [style=none] (5) at (-1.5, 0.75) {};
		\node [style=none] (8) at (1.5, 2.25) {};
		\node [style=none] (9) at (2.5, 1.5) {};
		\node [style=none] (10) at (1.5, 0.75) {};
		\node [style=none] (11) at (-1.5, 2) {};
		\node [style=none] (12) at (-2, 1.5) {};
		\node [style=none] (13) at (-1.5, 1) {};
		\node [style=none] (14) at (-1, 1.5) {};
		\node [style=none] (15) at (1.5, 2) {};
		\node [style=none] (16) at (1, 1.5) {};
		\node [style=none] (17) at (1.5, 1) {};
		\node [style=none] (18) at (2, 1.5) {};
		\node [style=green vertex] (19) at (1.5, -1.5) {};
		\node [style=cyan vertex] (20) at (-1.5, -1.5) {};
		\node [style=none] (21) at (2.5, -1.5) {};
		\node [style=none] (22) at (1.5, -2.25) {};
		\node [style=none] (23) at (1.5, -0.75) {};
		\node [style=none] (24) at (-1.5, -2.25) {};
		\node [style=none] (25) at (-2.5, -1.5) {};
		\node [style=none] (26) at (-1.5, -0.75) {};
		\node [style=none] (27) at (1.5, -2) {};
		\node [style=none] (28) at (2, -1.5) {};
		\node [style=none] (29) at (1.5, -1) {};
		\node [style=none] (30) at (1, -1.5) {};
		\node [style=none] (31) at (-1.5, -2) {};
		\node [style=none] (32) at (-1, -1.5) {};
		\node [style=none] (33) at (-1.5, -1) {};
		\node [style=none] (34) at (-2, -1.5) {};
		\node [style=none] (35) at (-1.5, 2.5) {};
		\node [style=none] (36) at (-0.75, 1.5) {};
		\node [style=none] (37) at (-2.25, 1.5) {};
		\node [style=none] (38) at (-0.75, -1.5) {};
		\node [style=none] (39) at (-1.5, -2.5) {};
		\node [style=none] (40) at (-2.25, -1.5) {};
		\node [style=none] (41) at (1.5, 2.5) {};
		\node [style=none] (42) at (2.25, 1.5) {};
		\node [style=none] (43) at (0.75, 1.5) {};
		\node [style=none] (44) at (2.25, -1.5) {};
		\node [style=none] (45) at (1.5, -2.5) {};
		\node [style=none] (46) at (0.75, -1.5) {};
		\node [style=none] (47) at (-3.5, 2.25) {};
		\node [style=none] (48) at (-4, 2.25) {$0_a$};
		\node [style=none] (49) at (4, -2.25) {$1_a$};
		\node [style=none] (50) at (3.5, -2.25) {};
		\node [style=none] (51) at (3.5, 2.25) {};
		\node [style=none] (52) at (4, 2.25) {$0_b$};
		\node [style=none] (53) at (-3.5, -2.25) {};
		\node [style=none] (54) at (-4, -2.25) {$1_b$};
		\node [style=none] (56) at (-4.25, 0) {$\mathsf{solo}$};
		\node [style=none] (57) at (-3.5, 0.25) {};
		\node [style=none] (58) at (-3.5, -0.25) {};
		\node [style=none] (59) at (-2, 1.25) {};
		\node [style=none] (60) at (-2, -1.25) {};
		\node [style=none] (61) at (4.25, 0) {$\mathsf{solo}$};
		\node [style=none] (62) at (2, 1.25) {};
		\node [style=none] (63) at (2, -1.25) {};
		\node [style=none] (64) at (3.5, 0.25) {};
		\node [style=none] (65) at (3.5, -0.25) {};
	\end{pgfonlayer}
	\begin{pgfonlayer}{edgelayer}
		\draw [style=hyperedge] (5.center)
			 to (10.center)
			 to [in=-90, out=0] (9.center)
			 to [in=0, out=90] (8.center)
			 to (4.center)
			 to [in=90, out=-180] (3.center)
			 to [in=-180, out=-90] cycle;
		\draw [style=hyperedge] (14.center)
			 to [in=0, out=90] (11.center)
			 to [in=90, out=-180] (12.center)
			 to [in=-180, out=-90] (13.center)
			 to [in=-90, out=0] cycle;
		\draw [style=hyperedge] (18.center)
			 to [in=0, out=90] (15.center)
			 to [in=90, out=-180] (16.center)
			 to [in=-180, out=-90] (17.center)
			 to [in=-90, out=0] cycle;
		\draw [style=hyperedge] (23.center)
			 to (26.center)
			 to [in=90, out=180] (25.center)
			 to [in=180, out=-90] (24.center)
			 to (22.center)
			 to [in=-90, out=0] (21.center)
			 to [in=0, out=90] cycle;
		\draw [style=hyperedge] (30.center)
			 to [in=180, out=-90] (27.center)
			 to [in=-90, out=0] (28.center)
			 to [in=0, out=90] (29.center)
			 to [in=90, out=180] cycle;
		\draw [style=hyperedge, in=180, out=-90] (34.center) to (31.center);
		\draw [style=hyperedge, in=-90, out=0] (31.center) to (32.center);
		\draw [style=hyperedge, in=0, out=90] (32.center) to (33.center);
		\draw [style=hyperedge, in=90, out=180] (33.center) to (34.center);
		\draw [style=hyperedge] (37.center) to (40.center);
		\draw [style=hyperedge, in=180, out=-90] (40.center) to (39.center);
		\draw [style=hyperedge, in=-90, out=0] (39.center) to (38.center);
		\draw [style=hyperedge] (38.center) to (36.center);
		\draw [style=hyperedge, in=0, out=90] (36.center) to (35.center);
		\draw [style=hyperedge, in=90, out=180] (35.center) to (37.center);
		\draw [style=hyperedge] (43.center) to (46.center);
		\draw [style=hyperedge, in=180, out=-90] (46.center) to (45.center);
		\draw [style=hyperedge, in=-90, out=0] (45.center) to (44.center);
		\draw [style=hyperedge] (44.center) to (42.center);
		\draw [style=hyperedge, in=0, out=90] (42.center) to (41.center);
		\draw [style=hyperedge, in=90, out=180] (41.center) to (43.center);
		\draw [style=gray arrow] (47.center) to (0);
		\draw [style=gray arrow] (53.center) to (20);
		\draw [style=gray arrow] (50.center) to (19);
		\draw [style=gray arrow] (51.center) to (1);
		\draw [style=gray arrow] (57.center) to (59.center);
		\draw [style=gray arrow] (58.center) to (60.center);
		\draw [style=gray arrow] (64.center) to (62.center);
		\draw [style=gray arrow] (65.center) to (63.center);
	\end{pgfonlayer}
\end{tikzpicture}
\end{center}

Let $v$ be the top-left vertex, where agent~$a$ has input value~$0$. Then by definition $H,v \models 0_a$.
Moreover, $a$ does not know whether agent~$b$ is alive: $H,v \models \neg \forallworld_a \existsagent_b \true$, which we could also reformulate as $H, v \models \neg \forallworld_a \mathsf{solo}$, that is, $a$ does not know whether this is a solo execution.
We can also say that $a$ considers possible that~$b$ is alive with value~$1$: $H,v \models \existsworld_a \existsagent_b 1_b$.
As a last example, we can express that $a$ knows that if $b$ is alive, its value is either $0$ or $1$: $H,v \models \forallworld_a (\neg \mathsf{solo} \Rightarrow \existsagent_b (0_b \lor 1_b))$.
Alternatively, we could have used the operator $\forallagent_b$ to express the same fact without a conditional: $H,v \models \forallworld_a \forallagent_b (0_b \lor 1_b)$.
\end{example}

\subsection{Safe and unsafe knowledge}
\label{sec:safe-knowledge}

In traditional epistemic logics, formulas are interpreted in a world of the model.
In the logic $\twoCH$, to talk about the knowledge of an agent in a world, we first need to quantify over the points of view of this agent.
Since there are two quantifiers $\existsagent_a$ and $\forallagent_a$, we obtain two different knowledge operators on worlds, which we call \emph{safe} and \emph{unsafe} knowledge.
\begin{mathpar}
\Ksafe_a \Phi := \existsagent_a \forallworld_a \Phi
\and
\Kunsafe_a \Phi := \forallagent_a \forallworld_a \Phi
\end{mathpar}

These two operators only differ in the knowledge of dead agents.
Indeed, given a world~$e$ and an agent~$a$ which is dead in~$e$ (i.e., $\proj_a(e)$ is undefined), we have $H,e \not\models \Ksafe_a \Phi$ (dead agents know nothing), whereas $H,e \models \Kunsafe_a \Phi$ (dead agents know everything).
However, when the agent~$a$ is alive in~$e$, the two notions agree: $H,e \models \Ksafe_a \Phi \iff H,e \models \Kunsafe_a \Phi$.

\section{Axiomatics}
\label{sec:prooftheory}


The logic $\twoCH$ has all the usual inference rules of classical propositional logic (such as modus ponens), as well as the following rules for modalities.
We annotate the $\vdash$ symbol with the sort of the corresponding formula, $a\in\cA$ for agent sorts and $e$ for the world sort.
\begin{mathpar}
    \inferrule*[Right=Nec-a]
        {\vdash_e \Phi}
        {\vdash_a\forallworld_a\Phi}
    \and
    \inferrule*[Right=Nec-e]
        {\vdash_a\varphi}
        {\vdash_e\forallagent_a\varphi}
    \and 
    \inferrule*[Right=RM]
        {\vdash_e \Phi \to \Psi}
        {\vdash_a \heartsuit\Phi \to \heartsuit\Psi}
    \and
    \inferrule*[Right=RM']
		{\vdash_a \phi \to \psi}
		{\vdash_e \heartsuit\phi \to \heartsuit\psi}
\end{mathpar}
\begin{mathpar}\mprset{fraction={===}}
    \inferrule*[Right=Adj-1]
        {\vdash_e\Phi \to \forallagent_a\psi}
        {\vdash_a\existsworld_a\Phi \to \psi}
    \and
    \inferrule*[Right=Adj-2]
        {\vdash_a\varphi \to \forallworld_a\Psi}
        {\vdash_e\existsagent_a\varphi \to \Psi}
\end{mathpar}
where $\heartsuit \in \{ \existsagent_a, \forallagent_a \}$ for rule RM, and $\heartsuit \in \{ \existsworld_a, \forallworld_a \}$ for rule RM'.
The first two rules are necessitation rules.
The next two rules are monotonicity rules.
The last two rules are called adjunction rules, and the double horizontal bar indicates that they go in both directions: top-to-bottom and bottom-to-top.
They describe the interaction between the two pairs of modalities.
Finally, we have the following axiom schemes for modalities:
\begin{itemize}
    \item $\vdash_a \varphi \to \existsworld_a\existsagent_a\varphi$ : every point of view belongs to some world;
    \item $\vdash_a \existsworld_a\existsagent_a\varphi \to \varphi$ : every world has at most one point of view of a given agent;
    \item $\vdash_e \bigvee_{a\in \mathcal{A}} \existsagent_a \true$ : every world contains at least one point of view.
\end{itemize}

As we will see, the axiom schemes for modalities correspond to the defining properties of chromatic hypergraphs.
Therefore, we will call the first axiom scheme \emph{surjecitivity}, the second one \emph{functionality}, and the third one \emph{non-emptiness}.

\begin{proposition}\label{prop:soundness}
    The logic $\twoCH$ is sound with respect to chromatic hypergraphs.
\end{proposition}

\subsection{Playing with logic $\twoCH$}

In this logic, we can show that the universal modalities satisfy axiom $\mathsf{K}$:

\begin{proposition}\label{prop:K-universal}
    For $\heartsuit\in \{\forallagent_a, \forallworld_a\}$, the axiom $\mathsf{K}_\heartsuit$ holds:
    $\heartsuit(\varphi \to \psi) \to (\heartsuit\varphi \to \heartsuit \psi)$.
\end{proposition}

We give a list of useful formulas that are derivable in the logic $\twoCH$.

\begin{proposition}\label{prop:useful-formulas}
    The following statements are derivable in the two-level logic $\twoCH$:
    \begin{multicols}{2}
    \begin{enumerate}
        \item $\existsagent_a\varphi \to \forallagent_a\varphi$
        \item $\forallworld_a\Phi \to \forallworld_a\existsagent_a\forallworld_a\Phi$
        \item $\existsagent_a\forallworld_a\Phi\to \Phi$
        \columnbreak
        \item $\Phi \to \forallagent_a\existsworld_a \Phi$
        \item $\varphi \to \forallworld_a\existsagent_a \varphi$
        \item $\forallworld_a\Phi\to \existsworld_a\Phi$
    \end{enumerate}
    \end{multicols}
\end{proposition}

\noindent
Here is an intuitive explanation of these formulas:
\begin{enumerate}
\item If $\phi$ holds in a point of view of~$a$, then it holds in all points of view of~$a$. That is, there can be at most one point of view per world.
\item This is a form of positive introspection: if an agent knows something, then he knows that he knows it.
However, we can put stress on the last ``he'', that is, ``if he knows something, then he knows that \emph{he} knows it''.
\item This is a form of veracity: if a fact about a world is known by someone, then it is true.
\item If a certain fact holds in a world, then the agents in this world consider this fact possible.
This is related to negative introspection.
\item Agents know the local facts about themselves.
\item This is the usual modal axiom $\mathsf{D}$. It reflects the fact that every point of view belongs to at least one world.
\end{enumerate}

In~\cite{gandalf-journal}, the use of local atomic propositions leads to a so-called assumption of locality, $K_a(p_{a,x}) \lor K_a(\neg p_{a,x})$.
In hypergraph models, valuations are local by construction:

\begin{proposition}
    For any $p_a \in \AP_a$, the formula $\forallworld_a\existsagent_a p_a \lor \forallworld_a\existsagent_a \neg p_a$ is derivable in $\twoCH$.
\end{proposition}

\begin{proof}
    By adjunction rules, we have $\existsworld_a\forallagent_a p_a \to p_a$ and $p_a \to \forallworld_a\existsagent_a p_a$.
    By cut rule, we have $\existsworld_a\forallagent_a p_a \to \forallworld_a\existsagent_a p_a$.
    This is equivalent to $\forallworld_a\existsagent_a p_a \lor \forallworld_a\existsagent_a \neg p_a$ by propositional logic.
\end{proof}
                    
Note that the above proof does not use the fact that $p_a$ is atomic, so in fact, an agent can decide \emph{any} formula about itself. That is, for any $\varphi$, $\forallworld_a\existsagent_a \varphi \lor \forallworld_a\existsagent_a \neg \varphi$ is derivable.

\subsection{Completeness}

The proof of completeness uses the standard canonicity argument, extended to the many-sorted case.
There is nothing surprising: the canonical model consists of the maximal consistent sets of formulas, now of several sorts.
These sets satisfy standard properties, and together form a chromatic hypergraph.

\begin{definition}
    A set of formulas $S_\ast$ of sort $\ast$ is \emph{inconsistent} if $\false$ can be derived from it.
    Otherwise, it is called consistent.
    A consistent set of formulas is \emph{maximal} if it is not a proper subset of any other consistent set of formulas.
\end{definition}

\begin{proposition}
    Maximal consistent sets (MCS) of formulas of sort $\ast$ are closed under modus ponens: for any formula $\xi$ of sort $\ast$, either $\xi\in S_\ast$ or $\neg\xi\in S_\ast$, and every (non-maximal) consistent set of formulas of sort $\ast$ is contained in a maximal consistent set of formulas of sort $\ast$.
\end{proposition}

\begin{definition}
    The canonical hypergraph model consists of the following:
    \begin{itemize}
        \item the set of hyperedges is $E = \{S_e\mid S_e \text{ is a MCS of sort }e\}$;
        \item for each agent $a$, the set of vertices is $V_a = \{S_a\mid S_a \text{ is a MCS of sort }a\}$;
    \end{itemize} 
    together with relations $R_a\subseteq E\times V_a$ for every agent $a$, which are defined as follows: $S_E R_a S_a$ iff for all formulas $\Phi$, if $\Phi\in S_E$, then $\existsworld_a \Phi\in S_a$.
    The valuation function is defined by: $\val_\ast(p_\ast) = \{S_\ast\mid p_\ast\in S_\ast\}$.
\end{definition}

\begin{proposition}\label{prop:rel-can-model}
    In the canonical model, if $S_E R_a S_a$, if $\forallworld_a\Phi \in S_a$, then $\Phi\in S_E$.
\end{proposition}

\begin{lemma}
    The canonical model is a chromatic hypergraph, that is, $R_a$ is surjective, functional, and for any $S_E$ there is $S_a$ such that $S_E R_a S_a$ for at least one $a$.
\end{lemma}
\begin{proof}
    First, suppose that $S_E R_a S_a$, $S_E R_a S'_a$, and $S_a\not=S'_a$.
    It means that there is a formula $\varphi$ which is in $S_a$, but is not in $S'_a$.
    By adjoint axiom and modus ponens, $\forallworld_a\existsagent_a\varphi$ is in $S_a$.
    By \Cref{prop:rel-can-model}, $\existsagent_a\varphi$ is in $S_E$.
    Using the definition of the canonical model, $\existsworld_a\existsagent_a\varphi$ belongs to~$S'_a$.
    From there, by functionality axiom and modus ponens, $\varphi\in S'_a$, which is a contradiction.
    Thus, $S_a = S'_a$.

    Second, we need to show that in the canonical model every vertex belongs to a hyperedge.
    Assume this is not the case, that is, there is a vertex $S_a$ that does not belong to any hyperedge.
    It means that there is no maximal consistent set of formulas that contains $S = \{\Phi\ |\ \forallworld_a\Phi\in S_a\}$.
    In particular, it means that $S$ is itself not consistent, that is, there is a finite set of formulas $\{\Phi_i\}$ such that $\bigwedge_i \Phi_i \to \false$ is derivable.
    By applying necessitation, we get that $\forallworld_a(\bigwedge_i \Phi_i \to \false)$ is derivable, thus belongs to $S_a$.
    As $\forallworld_a$ distributes over conjunction, and every $\forallagent_a \Phi_i$ is in $S_a$, we get that $\forallworld_a\bigwedge_i \Phi_i$ is in $S_a$.
    Applying modus ponens, we get that $\forallworld_a\false$ is in $S_a$.
    Using surjectivity axiom, we get that $\false$ is in $S_a$, that is $S_a$ is not consistent, which is a contradiction.
    
    Lastly, we need to show that every hyperedge contains some vertex.
    Suppose it is not the case, that is there is a hyperedge $S_E$ that does not contain any vertex.
    It means that for every $a$, the set $\existsworld_a S_E$ is not consistent.
    Thus, for all $a$, there is a finite set of formulas $\{\Phi^a_i\}\subset S_E$, such that $\bigwedge_i \existsworld_a\Phi^a_i \to \false$ is derivable.
    We now show that this implies that $S_E$ is not consistent.
    By applying necessitation, we get that $\forallagent_a(\bigwedge_i \existsworld_a\Phi^a_i \to \false)$ is derivable for every $a$.
    By (K) and modus ponens, we derive $\forallagent_a(\bigwedge \existsworld_a\Phi^a_i) \to \forallagent_a\false$.
    Combining them all together, we have that $\bigwedge_a\bigwedge_i(\forallagent_a\existsworld_a\Phi^a_i) \to \bigwedge_a\forallagent_a\false$ is derivable too.
    The antecedent is in $S_E$ because every $\Phi^a_i$ is in $S_E$ and $\Phi \to \forallagent_a\existsworld_a\Phi$ is an axiom.
    Thus, $\bigwedge_a\forallagent_a\false$ is in $S_E$, which means that $S_E$ is not consistent since $\bigvee_a\existsagent_a \true_a$ is an axiom, which is its negation. 
    We have a contradiction, which means that every hyperedge contains some vertex.   
\end{proof}

\begin{lemma}
    In the canonical model, $S_\ast \models_\ast \xi$ iff $\xi\in S_\ast$.
\end{lemma}

\begin{theorem}
	\label{thm:completeness}
    The logic $\twoCH$ is complete with respect to chromatic hypergraph models.
\end{theorem}

\section{Links to related work}
\label{sec:equivalence}

\subsection{Equivalence with partial epistemic frames}
\label{sec:equivalence-frames}

Let us recall first the definition of a partial epistemic frame, which has been one of the main models used in the study of epistemic logics such as $\KBfour$ in \cite{GoubaultLR21kb4}: 

\begin{definition}
    Given the set of agents $\mathcal{A}$,
    a partial epistemic frame $\mathcal{M}$ consists of a set of worlds $M$
    together with a family of partial equivalence relations $\{\sim_a\}_{a\in \mathcal{A}}$, such that for every $w\in M$, $w\sim_a w$ for at least one $a\in \mathcal{A}$.
    A morphism of partial epistemic frames is a function $f:M\to M'$ such that for every $a\in \mathcal{A}$ and $w,w'\in M$, $w\sim_a w'$ implies $f(w)\sim_a f(w')$.
\end{definition}

We can transform a partial epistemic frame into a chromatic hypergraph, and vice versa, using the following construction.
Suppose we are given a partial epistemic frame $\mathcal{M}$.
We construct a chromatic hypergraph $\eta(\mathcal{M})$ by setting $E=M$ and $V_a = M/_{\sim_a}$, that is the set of hyperedges is exactly the set of worlds, and the set of vertices of color $a$ is the set of equivalence classes of $\sim_a$.
We then set $\proj_a(w)$ to be $[w]_a$, that is the equivalence class of $w$ under $\sim_a$.
It is easy to check that this indeed defines a chromatic hypergraph.

Conversely, given a chromatic hypergraph $H$, we can construct a partial epistemic frame $\kappa(H)$.
We set the set of worlds $M$ to be equal to the set of hyperedges of $H$, and $e\sim_a e'$ if and only if $e$ and $e'$ share an $a$-colored vertex.
This yields a partial equivalence relation.

These maps can be seen as the dual hypergraph construction: if $H = (V, E)$ is a (non-chromatic) hypergraph, then $H^\ast$ is the hypergraph $(E, V)$ where the hyperedges are the vertices of $H$ and the vertices are the hyperedges of $H$.
Partial epistemic frames can be seen as hypergraphs that have colored hyperedges which are defined by equivalence classes.
The correspondence is exemplified in Figure~\ref{fig:frames-hypergraphs}.

\begin{figure}[h]
	\centering
        \begin{tikzpicture}
	\begin{pgfonlayer}{nodelayer}
		\node [style=green vertex] (0) at (0, 2) {};
		\node [style=cyan vertex] (1) at (-2, -1) {};
		\node [style=magenta vertex] (2) at (2, -1) {};
		\node [style=none] (3) at (2, -2) {};
		\node [style=none] (4) at (0, 3) {};
		\node [style=none] (5) at (3, -1) {};
		\node [style=none] (6) at (-1, 2) {};
		\node [style=none] (7) at (0, 3) {};
		\node [style=none] (8) at (-2, -2) {};
		\node [style=none] (9) at (1, 2) {};
		\node [style=none] (10) at (-3, -1) {};
		\node [style=none] (11) at (2.5, -1.5) {};
		\node [style=none] (12) at (-2.5, -0.5) {};
		\node [style=none] (13) at (2.5, -0.5) {};
		\node [style=none] (14) at (-2.5, -1.5) {};
		\node [style=none] (15) at (3, -2) {};
		\node [style=none] (16) at (-2, 3) {};
		\node [style=none] (17) at (2, 3) {};
		\node [style=none] (18) at (-3, -2) {};
		\node [style=none] (19) at (0, 2.5) {};
		\node [style=none] (20) at (-0.5, 2) {};
		\node [style=none] (21) at (0, 1.5) {};
		\node [style=none] (22) at (0.5, 2) {};
		\node [style=none] (23) at (-2, -0.5) {};
		\node [style=none] (24) at (-2.5, -1) {};
		\node [style=none] (25) at (-2, -1.5) {};
		\node [style=none] (26) at (-1.5, -1) {};
		\node [style=none] (27) at (2, -0.5) {};
		\node [style=none] (28) at (1.5, -1) {};
		\node [style=none] (29) at (2, -1.5) {};
		\node [style=none] (30) at (2.5, -1) {};
	\end{pgfonlayer}
	\begin{pgfonlayer}{edgelayer}
		\draw [style=hyperedge] (3.center)
			 to [in=-75, out=-15] (5.center)
			 to [in=-15, out=105, looseness=0.25] (4.center)
			 to [in=105, out=165] (6.center)
			 to [in=165, out=-75, looseness=0.25] cycle;
		\draw [style=hyperedge] (7.center)
			 to [in=75, out=15] (9.center)
			 to [in=15, out=-105, looseness=0.25] (8.center)
			 to [in=-105, out=-165] (10.center)
			 to [in=-165, out=75, looseness=0.25] cycle;
		\draw [style=hyperedge] (11.center)
			 to [in=-45, out=45, looseness=1.25] (13.center)
			 to [in=45, out=135, looseness=0.50] (12.center)
			 to [in=135, out=-135, looseness=1.25] (14.center)
			 to [in=-135, out=-45, looseness=0.50] cycle;
		\draw [style=hyperedge] (18.center)
			 to [in=-135, out=-45, looseness=0.75] (15.center)
			 to [in=-30, out=45, looseness=0.75] (17.center)
			 to [in=30, out=150] (16.center)
			 to [in=135, out=-150, looseness=0.75] cycle;
		\draw [style=hyperedge] (22.center)
			 to [in=0, out=90] (19.center)
			 to [in=90, out=-180] (20.center)
			 to [in=-180, out=-90] (21.center)
			 to [in=-90, out=0] cycle;
		\draw [style=hyperedge] (26.center)
			 to [in=0, out=90] (23.center)
			 to [in=90, out=-180] (24.center)
			 to [in=-180, out=-90] (25.center)
			 to [in=-90, out=0] cycle;
		\draw [style=hyperedge] (30.center)
			 to [in=0, out=90] (27.center)
			 to [in=90, out=-180] (28.center)
			 to [in=-180, out=-90] (29.center)
			 to [in=-90, out=0] cycle;
	\end{pgfonlayer}
\end{tikzpicture}
        \qquad
        \begin{tikzcd}
            {} & {} & {} \\
	        \arrow["\eta", shift left=4, from=1-3, to=1-1]
            \arrow["\kappa", shift left=4, from=1-1, to=1-3]
        \end{tikzcd}
        \qquad
        \begin{tikzpicture}
	\begin{pgfonlayer}{nodelayer}
		\node [style=white vertex] (0) at (-2, 0) {};
		\node [style=white vertex] (1) at (-1, -0.75) {};
		\node [style=white vertex] (2) at (-3, -2) {};
		\node [style=white vertex] (3) at (1, -2) {};
		\node [style=white vertex] (4) at (0, 0) {};
		\node [style=white vertex] (5) at (-1, -2) {};
		\node [style=white vertex] (6) at (-1, 2) {};
		\node [style=none] (41) at (-1, 2.75) {};
		\node [style=none] (66) at (-2.325, 0.6) {};
		\node [style=none] (67) at (-3.9, -2.55) {};
		\node [style=none] (68) at (-0.775, -2.5) {};
		\node [style=none] (69) at (-0.525, -0.425) {};
		\node [style=none] (70) at (0.3, 0.6) {};
		\node [style=none] (71) at (1.875, -2.55) {};
		\node [style=none] (72) at (-1.25, -2.5) {};
		\node [style=none] (73) at (-1.5, -0.425) {};
		\node [style=none] (78) at (-1, 2.75) {};
		\node [style=none] (79) at (-2.5, 0.25) {};
		\node [style=none] (80) at (-1, -1.25) {};
		\node [style=none] (81) at (0.5, 0.25) {};
	\end{pgfonlayer}
	\begin{pgfonlayer}{edgelayer}
		\draw [style=fill cyan] (69.center)
			 to [in=30, out=120, looseness=0.75] (66.center)
			 to [in=120, out=-150, looseness=0.50] (67.center)
			 to [in=-150, out=-60, looseness=0.50] (68.center)
			 to [in=-60, out=30, looseness=0.75] cycle;
		\draw [style=fill magenta] (73.center)
			 to [in=150, out=60, looseness=0.75] (70.center)
			 to [in=60, out=-30, looseness=0.50] (71.center)
			 to [in=-30, out=-120, looseness=0.50] (72.center)
			 to [in=-120, out=150, looseness=0.75] cycle;
		\draw [style=fill green] (79.center)
			 to [in=-180, out=-90, looseness=0.75] (80.center)
			 to [in=-90, out=0, looseness=0.75] (81.center)
			 to [in=0, out=90, looseness=0.50] (78.center)
			 to [in=90, out=-180, looseness=0.50] cycle;
	\end{pgfonlayer}
\end{tikzpicture}
    \caption{Example of correspondence between chromatic hypergraphs and frames.}
    \label{fig:frames-hypergraphs}
\end{figure}
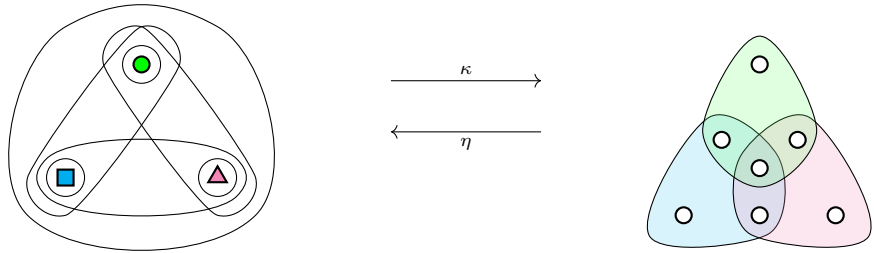

In fact, $\eta$ and $\kappa$ can be extended to morphisms of partial epistemic frames and chromatic hypergraphs, giving an equivalence of categories.
First, we need to define the corresponding morphisms of chromatic hypergraphs:

\begin{definition}
	A morphism of hypergraphs $f: H\to H'$ is a family of functions $f_a: V_a\to V'_a$ for each agent $a$, together with a function $f_e: E\to E'$ such that for all $a\in \mathcal{A}$, if $\proj_a(e)=v$ then $\proj'_a(f_e(e)) = f_a(v)$.
\end{definition}



\begin{theorem}\label{prop:kripke-hypergraph}
    The category of partial epistemic frames is isomorphic to the category of chromatic hypergraphs.
    In particular, for any chromatic hypergraph $H$, $\eta(\kappa(H))$ is isomorphic to $H$, and for any partial epistemic frame $\mathcal{M}$, $\kappa(\eta(\mathcal{M}))$ is isomorphic to $\mathcal{M}$.
\end{theorem}

In light of \Cref{prop:kripke-hypergraph}, chromatic hypergraphs and partial epistemic frames contain exactly the same information.
So, in theory, we could have defined the semantics of $\twoCH$ in partial epistemic frames.
However, this would be quite unnatural to do, since epistemic frames do not have a tangible notion of point of view: we would have to attach atomic propositions, and interpret agent formulas, in the equivalence classes of~$\sim_a$.
%

Instead, we can still embed partial epistemic models (with only world atomic propositions) into a subclass of chromatic hypergraphs models, such that $\AP_a = \varnothing$ for every agent.
From this, we can extend \Cref{prop:kripke-hypergraph} to work at the level of models:


\begin{corollary}\label{prop:kripke-hypergraph-2}
    The category of partial epistemic models is isomorphic to the category of chromatic hypergraph models with empty sets of atomic propositions for agents.
\end{corollary}

\subsection{Translation from $\KBfour$ to $\twoCH$}

We can use the equivalence of epistemic frames and hypergraphs for showing how the logics $\KBfour$ and $\twoCH$ are related: we will show that $\twoCH$ is a conservative extension of $\KBfour + \mathsf{NE}$, where axiom $\mathsf{NE}$ ensures that there is an alive agent in each world (see \cite{GoubaultLR21kb4} for further details).
From semantics side, the worlds of epistemic frames are the hyperedges of hypergraphs, thus the formulas of $\KBfour$ are to be translated to the \emph{world} formulas of $\twoCH$.
In particular, when defining the translation, we set the set of world atomic propositions to be the set of atomic propositions of $\KBfour$.
The translation of formulas is defined recursively as follows:
\begin{mathpar}
\translate{p} \;:=\; p \and
\translate{\neg \Phi} \;:=\; \neg \translate{\Phi} \and
\translate{\Phi \land \Psi} \;:=\; \translate{\Phi} \land \translate{\Psi} \and
\translate{K_a \Phi} \;:=\; \forallagent_a\forallworld_a \translate{\Phi}
\end{mathpar}
Essentially, this translation interprets the knowledge operator of $\KBfour$ using the \emph{unsafe knowledge} operator described in \cref{sec:safe-knowledge}.
So, if in a given world agent~$a$ is dead, $\translate{K_a \Phi}$ will be vacuously true.

\begin{proposition}\label{prop:translation-validity}
    For a partial epistemic frame $\mathcal{M}$ and a formula $\Phi$ of $\KBfour$, $\mathcal{M}, w\models \Phi$ iff $\eta(\mathcal{M}), \eta(w)\models_e \translate{\Phi}$.
\end{proposition}
\begin{proof}
    We show the statement by induction on the structure of $\Phi$.
    For atomic propositions, as well as boolean connectives, the proof is trivial.
    For the modality, we have: $\mathcal{M}, w\models K_a\Phi$ if and only if for all $w'\in M$ such that $w\sim_a w'$, $\mathcal{M}, w'\models \Phi$.
    By induction, this is equivalent to for all $w'\in M$ such that $w\sim_a w'$, $\eta(\mathcal{M}), \eta(w')\models \translate{\Phi}$.
    By definition of $\eta$, it is the same as for all hyperedges $e\in \eta(\mathcal{M})_E$ that share an $a$ vertex with $\eta(w)$, $\eta(\mathcal{M}), e\models \translate{\Phi}$.
    This is equivalent to $\eta(\mathcal{M}), \eta(w)\models \translate{K_a\Phi}$.  
\end{proof}

\begin{corollary}\label{cor:translation-validity}
    $\Phi$ is valid in a partial epistemic frame $\mathcal{M}$ iff $\translate{\Phi}$ is valid in $\eta(\mathcal{M})$.
\end{corollary}

Using this and \cref{thm:completeness}, we can show that $\twoCH$ is a conservative extension of $\KBfour$:

\begin{theorem}
	\label{thm:conservative-extension}
    For every $\KBfour$-formula $\Phi$, $\vdash_{\KBfour + \mathsf{NE}} \Phi$ if and only if $\vdash_e \translate{\Phi}$.
\end{theorem}
\begin{proof}
    By completeness, for $\KBfour$ we have $\vdash_{\KBfour + \mathsf{NE}} \Phi \Leftrightarrow\ \models \Phi$. 
    By \Cref{cor:translation-validity}, we have $\models \Phi \Leftrightarrow\ \models_e \translate{\Phi}$.
    And by completeness for $\twoCH$, we have $\models_e \translate{\Phi} \Leftrightarrow\ \vdash_e \translate{\Phi}$.
\end{proof}

As a corollary, we get that the combined modality $\forallagent_a\forallworld_a$ satisfies axioms $\mathsf{K}$, $\mathsf{B}$, and $\mathsf{4}$.

\begin{remark}
	Similarly, one can wonder which logic we would get if we translate formulas using instead the \emph{safe knowledge} operator $\Ksafe_a \Phi = \existsagent_a\forallworld_a\Phi$.
	First thing to note is that this does \emph{not} yield the three-valued logic $\Sfive^{\bowtie}$ of~\cite{Ditmarsch22Complete}\footnote{A translation of $\Sfive^{\bowtie}$ into $\twoCH$ is possible, but it is more involved. One must first translate the well-definedness judgment, then set $\translate{K_a \Phi} := \Ksafe_a (\text{well-defined}(\Phi) \Rightarrow \translate{\Phi})$.}.
	We can show that the safe knowledge modality satisfies axioms $\mathsf{K}$, $\mathsf{T}$, and $\mathsf{B}$, but not $\mathsf{4}$.
	However, this modality is not normal as the necessitation rule is not admissible:
	$\true$ is valid in every world, but it is not the case that $\existsagent_a \forallworld_a \true$ is valid in every world.
	Thus, we cannot derive $\existsagent_a \forallworld_a \true$ from $\true$.
\end{remark}

\subsection{Correspondence with neighborhood frames}
\label{sec:neighborhood}

Chromatic hypergraphs require that every hyperedge contains at most one vertex of each color.
So in every world, an agent can have either~$0$ or~$1$ point of view.
But what happens if we drop this condition and allow agents to have multiple points of view about a given world?
Technically, this can be achieved by replacing, in \cref{def:chr-hypergraph}, the partial function $\proj_a : E \to V_a$ by a relation $\proj_a \subseteq E \times V_a$, i.e., get rid of the \emph{functionality} requirement.

This leads to an intriguing connection with neighborhood frames \cite{pacuit2017neighborhood}.
In this subsection we will not prove formal results, but rather give an intuition of the connection between the two notions, applying a similar construction as in \cref{sec:equivalence-frames}.

Neighborhood frames generalize epistemic frames by allowing agents to have multiple points of view on the same world, which on the side of hypergraphs corresponds exactly to the situation when we allow hyperedges to contain multiple vertices of the same color:

\begin{definition}[\cite{pacuit2017neighborhood}]
    A neighborhood frame is a pair $\mathsf{M} = (S, \{N_a\}_{a\in \mathcal{A}})$,
    where $S$ is a set of states,
    and for every agent $a \in \cA$,
    $N_a$ 
    is a function that assigns to every state $s \in S$, a set $N_a(s) \subseteq 2^S$ called the $a$-neighborhoods of $s$.
\end{definition}

An example of a situation where neighborhood frames are required is as follows.
Suppose we have two processes, communicating through shared memory.
The memory has two cells, and each cell stores a bit of information: $0$ or $1$.
Processes are given access to memory cells arbitrarily, and both can be assigned the same cell.
They know which cell is assigned to them, and they know the value that is stored in this cell, that is, they read the value of the cell.
Therefore, a process can have two points of view on the same situation, depending on which cell it is given access to.
Assume for example that the shared memory stores values~$(0, 1)$.
Process $a$, when assigned the first cell, knows that the memory stores $0$, and when assigned the second cell, knows that the memory stores $1$.
So the set of possible states is $S = \{ (0,0), (0,1), (1,0), (1,1) \}$.
In state $(0,1)$, the two possible points of view of process~$a$ are described by neighborhoods: $N_a((0,1)) = \{ \{ (0,0), (0,1) \}, \{ (0,1), (1,1) \} \}$

We can also make sense of this example using 
generalized chromatic hypergraphs, where we allow hyperedges that contain multiple points of view of the same agent.
In that case, we can model our example as follows.
We have two agents, $a$ and $b$.
There are four hyperedges corresponding to four possible states of the memory: $(0, 0)$, $(0, 1)$, $(1, 0)$, and $(1, 1)$.
Each agent has four possible points of view, depending on which memory cell is assigned (left or right), and which bit is read ($0$ or $1$).
Let us denote the vertices of agent~$a$ by $\{(a,m,b) \mid m \in \{L,R\}, b \in \{0,1\} \}$, and similarly for agent~$b$.
Then, a vertex $(a,m,b)$ belongs to a hyperedge $(x_L, x_R)$ if and only if $x_m = b$.


Recall the duality construction of \cref{fig:frames-hypergraphs}, switching the role of vertices and hyperedges.
In our case, the set of hyperedges becomes the set of states, and the set of vertices defines the neighborhood function.
In our example, the hyperedge/state $s = (0,a)$ contains two vertices of agent~$a$: $(a,L,0)$ and $(a,R,1)$.
The first vertex corresponds to the $a$-neighborhood $\{(0,0), (0,1)\}$, which is the set of hyperedges containing this vertex.
Similarly, the second vertex corresponds to the neighborhood $\{(0,1), (1,1)\}$.
This recovers the set $N_a(s)$ of the corresponding neighborhood frame.
Note that we do not obtain all neighborhood frames in this way, but only those in which a world belongs to all of its neighborhoods.

\section{Conclusion}

In this paper, we proposed a many-sorted modal logic for reasoning about knowledge in multi-agent systems that treat as first-class citizens both 
participating agents and the environment.
This allowed us to reconcile the numerous logics and models of the literature, which indeed struggled with expressing coherent general global properties of worlds and local properties of agents. 
There are two main extensions that we are currently studying based on this work.
First, having points of view of agents as first-class citizens, a combination of epistemic logics with temporal modalities allows us to provide a framework with greater emphasis on local action of agents, compared to e.g.,\ DEL \cite{hvdetal.DEL:2007} or interpreted systems \cite{fagin}. 
Secondly, a natural question arises as to whether we can reconcile chromatic hypergraphs, 
with chromatic (semi-)simplicial sets as studied in e.g.,~\cite{GoubaultKLR23semisimplicial}. This would allow us to naturally extend our logic with a distributed knowledge operator.

%


\providecommand{\noopsort}[1]{}

\appendix

\section{Proofs}

\subsection{Proof of \Cref{prop:K-universal}}

\begin{proof}
    First, we show that universal modalities distribute over conjunction, that is $\heartsuit(\xi \land \eta) \leftrightarrow (\heartsuit\xi \land \heartsuit\eta)$.
    Left-to-right direction: we have that $(\xi \land \eta) \to \xi$ and $(\xi \land \eta) \to \eta$. 
    Applying the RM rule, we get that $\heartsuit(\xi \land \eta) \to \heartsuit\xi$ and $\heartsuit(\xi \land \eta) \to \heartsuit\eta$.
    From this, the left-to-right direction follows.
    Right-to-left direction: denote the modality adjoint to $\heartsuit$ by $\spadesuit$, that is, if $\heartsuit=\forallagent_a$ then $\spadesuit=\existsworld_a$ and if $\heartsuit=\forallworld_a$ then $\spadesuit=\existsagent_a$.
    We have that $\spadesuit\heartsuit\xi \to \xi$ from $\heartsuit\xi \to \heartsuit\xi$ and the corresponding adjunction rule, similarly for $\eta$.
    From this we have that $\spadesuit\heartsuit\xi \land \spadesuit\heartsuit\eta \to \xi \land \eta$.
    By the same proof as in the left-to-right direction, we have that $\spadesuit(\heartsuit\xi \land \heartsuit\eta) \to \spadesuit\heartsuit\xi \land \spadesuit\heartsuit\eta$.
    Thus, we have that $\spadesuit(\heartsuit\xi \land \heartsuit\eta) \to \xi \land \eta$.
    Applying the adjunction rule, we get that $\heartsuit(\xi \land \eta) \to \heartsuit\xi \land \heartsuit\eta$, which is the right-to-left direction.

    The fact that $\mathsf{K}$ follows from the distribution of universal modalities over conjunction is a standard proof:
    From $((\xi\to \eta) \land \xi) \to \eta$ by RM we have $\heartsuit((\xi\to \eta) \land \xi) \to \heartsuit\eta$.
    By distribution, we have $(\heartsuit(\xi\to \eta) \land \heartsuit\xi) \to \heartsuit((\xi\to\eta)\land\xi)$. 
    Combining these two, we get $(\heartsuit(\xi\to \eta) \land \heartsuit\xi) \to \heartsuit\eta$, and thus $\heartsuit(\xi\to \eta) \to (\heartsuit\xi \to \heartsuit\eta)$.
\end{proof}

\subsection{Proof of \Cref{prop:useful-formulas}}

\begin{proof}
    Recall the list of formulas:

    \begin{multicols}{2}
        \begin{enumerate}
            \item $\existsagent_a\varphi \to \forallagent_a\varphi$;
            \item $\forallworld_a\Phi \to \forallworld_a\existsagent_a\forallworld_a\Phi$;
            \item $\existsagent_a\forallworld_a\Phi\to \Phi$;
            \columnbreak
            \item $\Phi \to \forallagent_a\existsworld_a \Phi$;
            \item $\varphi \to \forallworld_a\existsagent_a \varphi$;
            \item $\forallworld_a\Phi\to \existsworld_a\Phi$.
        \end{enumerate}
    \end{multicols}

    For the first formula, we just apply the adjunction rule to $\existsworld_a\existsagent_a\varphi \to\varphi$, which is an axiom. 
    In order to show the second formula, just apply the adjunction rule to $\existsagent_a\forallworld_a\Phi \to \existsagent_a\forallworld_a\Phi$, which is a tautology.
    Formulas 3, 4 and 5 are derived from $\forallworld_a\Phi \to \forallworld_a\Phi$, $\existsworld_a\Phi \to \existsworld_a\Phi$, and $\existsagent_a\varphi \to \existsagent_a\varphi$ respectively by applying the adjunction rule.
    The last formula is shown as follows: from formulas 3 and 4 we have $\existsagent_a\forallworld_a\Phi \to \forallagent_a\existsworld_a\Phi$. 
    Applying the adjunction rule, we get $\existsworld_a\existsagent_a\forallworld_a\Phi\to \existsworld_a\Phi$. 
    We also have $\forallworld_a\Phi \to \existsworld_a\existsagent_a\forallworld_a\Phi$, which is an axiom.
    From the last two formulas, we get $\forallworld_a\Phi \to \existsworld_a\Phi$.
\end{proof}

\subsection{Proof of \Cref{prop:kripke-hypergraph}}
\begin{proof}
    We just need to show how $\eta$ and $\kappa$ are extended to morphisms. 
    Functoriality is then straightforward, and checking that $\eta(\kappa(H)) \simeq H$ and $\kappa(\eta(\mathcal{M})) \simeq \mathcal{M}$ is also straightforward.
    Let $f: H\to H'$ be a morphism of chromatic hypergraphs.
    Then $\eta(f): \eta(H) \to \eta(H')$ just sends a world $e$ to $f_E(e)$.
    This is indeed a morphism of partial epistemic frames: suppose two worlds $e$ and $e'$ in $\eta(H)$ are $\sim_a$-equivalent.
    It means that in $H$ these two hyperedges share a vertex, and thus in $H'$ the two hyperedges $f_E(e)$ and $f_E(e')$ share a vertex, and thus $f_E(e) \sim_a f_E(e')$.

    Now let $g: \mathcal{M} \to \mathcal{M}'$ be a morphism of partial epistemic frames.
    Then $\kappa(g): \kappa(\mathcal{M}) \to \kappa(\mathcal{M}')$ sends a hyperedge $w$ to $g(w)$, thus $\kappa(g)_E$ is defined. 
    We need to show that it induces a map on vertices, and that the condition for morphsisms is satisfied.
    As $g$ preserves $\sim_a$, it induces a map on equivalence classes, which is exactly $\kappa(g)_a$.
    Let $w$ be a hyperedge and $v$ be a vertex in $\kappa(\mathcal{M})$, such that $\proj_a(w) = v$.
    It means that $w$ belongs to the equivalence class corresponding to $v$ in $\mathcal{M}$.
    Thus, g(w) belongs to the equivalence class corresponding to $g(v)$ in $\mathcal{M}'$, and thus $\proj_a(g(w)) = g(v)$.
\end{proof}



\end{document}